\documentclass[aps,
pra,
longbibliography,
reprint,
amsmath,showkeys,amssymb,superscriptaddress]{revtex4-2}

\usepackage[english]{babel}

\usepackage{bm}
\usepackage{braket}

\usepackage{hyperref}       
\hypersetup{
  colorlinks   = true, 
  urlcolor     = blue, 
  linkcolor    = blue, 
  citecolor   = blue 
}
\usepackage{cleveref}

\usepackage{tikz}
\usepackage{pgfplots}
\pgfplotsset{width=8cm,compat=1.18}
\usetikzlibrary{calc,angles,positioning,intersections,quotes,decorations.markings}
\usepackage{tkz-euclide}

\usepackage{amsthm}
\usepackage{enumerate}
\usepackage{enumitem,kantlipsum}

\newtheorem*{theorem*}{Theorem}
\newtheorem*{prop*}{Proposition}

\theoremstyle{definition}

\theoremstyle{remark}

\newcommand{\bb}{\begin{equation}}
\newcommand{\ee}{\end{equation}}

\newcommand{\id}{{\textnormal{id}}}
\newcommand{\tr}{{\textnormal{tr}}}

\newcommand{\cS}{\mathcal{S}}
\newcommand{\cB}{\mathcal{B}}

\usepackage{mathtools}  
\DeclareMathOperator{\Tr}{tr}

\begin{document}

\title{Fundamental Limit on the Power of Entanglement Assistance in Quantum Communication}

\author{Lasse H. Wolff}
\email{lhw@math.ku.dk}
\affiliation{Department of Mathematical Sciences, University of Copenhagen, Universitetsparken 5, 2100 Denmark}

\author{Paula Belzig}
\email{pbelzig@uwaterloo.ca}
\affiliation{Institute for Quantum Computing, University of Waterloo, 200 University Avenue West, 
Waterloo, Ontario N2L3G1, Canada}

\author{Matthias Christandl}
\email{christandl@math.ku.dk}
\affiliation{Department of Mathematical Sciences, University of Copenhagen, Universitetsparken 5, 2100 Denmark}

\author{Bergfinnur Durhuus}
\email{durhuus@math.ku.dk}
\affiliation{Department of Mathematical Sciences, University of Copenhagen, Universitetsparken 5, 2100 Denmark}

\author{Marco Tomamichel}
\email{marco.tomamichel@nus.edu.sg}
\affiliation{Department of Electrical and Computer Engineering, National University of Singapore, 4 Engineering Drive 3, 117583, Singapore  and\\
Centre for Quantum Technologies, 3 Science Drive 2, 117543, Singapore}

\date{\today}

\begin{abstract}
The optimal rate of reliable communication over a quantum channel can be enhanced by pre-shared entanglement. Whereas the enhancement may be unbounded in infinite-dimensional settings even when the input power is constrained, a long-standing conjecture asserts that the ratio between the entanglement-assisted and unassisted classical capacities is bounded in finite-dimensional settings [Bennett \emph{et al}., \emph{IEEE Trans. Inf. Theory} {\bf 48}, 2637 (2002)]. In this work, we prove this conjecture by showing that their ratio is upper bounded by $o(d^2)$, where $d$ is the input dimension of the channel. An application to quantum communication with noisy encoders and decoders is given.
\end{abstract}


\maketitle

\section{Introduction} \label{section:introduction}

The fact that pre-shared quantum entanglement between a sender and a receiver can lead to a boost in the amount of information that can be communicated over a quantum channel has long been a celebrated and widely studied phenomenon, which highlights the unique and novel possibilities that quantum phenomena inherently present for the development of new technology. The textbook example of this is the super-dense coding protocol \cite{Bennett_superdensecoding}, where in a noiseless setting, pre-shared entanglement allows a sender to communicate two bits of classical information by only transmitting a single qubit to the receiver, a fact that has been demonstrated experimentally \cite{MWKZ96}.

It turns out that, when communicating over a noisy channel, the advantage of using entanglement can become much more potent than for a noiseless channel, with even further practical implications \cite{Hao_practicalEACOMM}. For example, for a $d$-dimensional depolarising channel, the ratio between its \emph{entanglement-assisted classical capacity} $C_{\text{E}}$ and its (ordinary) \emph{classical capacity} $C$ approaches $d+1$ \cite{Bennet_1999_RatioCalculationDepolarising} as the depolarizing probability approaches $1$, i.e., as the channel approaches the fully depolarizing channel (see Fig.~\ref{figure:plot capacities depolarising channels}). This means that in the mentioned limit, the asymptotic rate at which classical bits can be reliably communicated over the channel with the help of quantum encoders and decoders will be $d+1$ times as large if the sender and receiver have access to an unlimited amount of pre-shared entangled qubits. Thus, entanglement can grant an arbitrarily large boost to the classical communication rate of a sufficiently noisy channel as the input and output space dimension grow to infinity---a phenomenon with applications in e.g. deep-space communication \cite{Banaszek_Qlimits, BELENCHIA20221, Nafria_23} and covert communication \cite{zlotnick2023, bash_quantum-secure_2015, Tahmasbi2020c}.

\begin{figure}[ht]

\centering

\includegraphics[width=8.6cm]{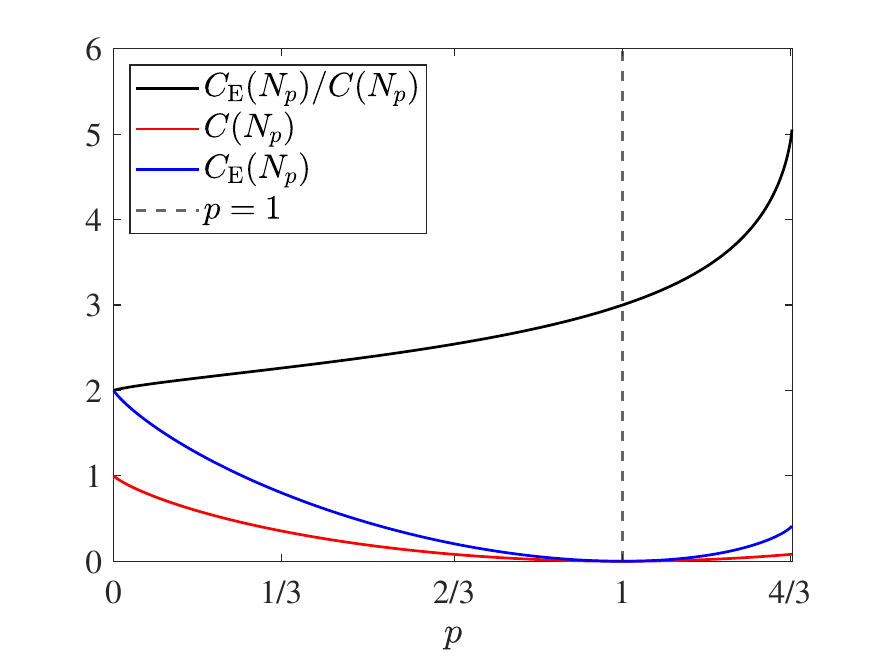}

\caption{The entanglement-assisted classical capacity $C_{\text{E}}(N_p)$, the classical capacity $C(N_p)$ and their ratio plotted as a function of $p$ for the qubit depolarising channel $N_p(\rho)  =(1-p) \rho + p \Tr(\rho) \frac{1}{2} \text{I}_2$, which is a quantum channel for $p \in \left[0,\frac{4}{3} \right]$. $C_{\text{E}}(N_p)$ and $C(N_p)$ are calculated analytically in \cite{Holevo_conjectureChannel} (see also 
 \cite[Theorem 1]{King_depolarising}). Note that $C_{\text{E}}(N_1)=C(N_1)=0$, where their ratio converges to $3$.} \label{figure:plot capacities depolarising channels}
\end{figure}

For other families of quantum channels, it has similarly been observed that the ratio $C_{\text{E}}/C$ is finite and bounded for finite input and output dimensions (see \cite{Bennet_1999_RatioCalculationDepolarising} for the erasure channel, \cite{Bennet_etAl_entanglementassisted_descript} for the amplitude damping channel) and can become arbitrarily large as the input and output space dimensions are unbounded (see \cite{Shi_2020} for the lossy bosonic channel with power constraints). 
It is therefore natural to ask whether it is possible for this ratio to diverge for families of channels with fixed input and output dimension, or, conversely, whether the ratio is bounded by a function of input and output dimension for all quantum channels. This question was, indeed asked in one of the earliest papers on entanglement assistance \cite{Bennet_etAl_entanglementassisted_descript}. In the present paper, we show that the latter alternative occurs by proving that the ratio $C_{\text{E}}/C$ is upper bounded by $O(d_{A}^2 / \ln(d_{A}))$, where $d_{A}$ is the input dimension of the channel. This is the content of the theorem below presented in Section \ref{section:presenting the bounds} and proven in Section \ref{section:proving the bounds}. In Section \ref{section:applications}, we discuss an application of our bound to fault-tolerant quantum communication.

\section{A Bound on \texorpdfstring{$C_{\text{E}}/C$}{CEoverC}} \label{section:presenting the bounds}

We consider Hilbert spaces $A$ and $B$ of finite dimension $d_A$ and $d_B$, respectively, whose unit vectors represent the pure states of quantum systems with the same label, and, as usual, we set $AB = A\otimes B$, etc. We shall denote the set of (bounded) linear maps on $A$ by $\cB(A)$ and by $\text{I}_A$ the identity operator on $A$. We let $\cS(A)$ denote the compact convex set of \emph{quantum states} (or density matrices), i.e., positive semi-definite elements of $\cB(A)$ with trace $1$. The set of pure states $\cS_*(A)$ consists of the extremal points of $\cS(A)$, being of the form $\rho_A = |\psi\rangle\langle\psi|$, where $|\psi\rangle$ is a unit vector in $A$.  A \emph{quantum channel} is a linear, completely positive, and  trace preserving map $T: \cB(A) \to \cB(B)$. By $\id_A$ we denote the identity channel on  $\cB(A)$.

Given a quantum channel $T:\mathcal{B}(A) \rightarrow \mathcal{B}(B)$ the classical capacity $C(T)$ is defined as the maximal asymptotic rate (in bits per channel use) at which classical information can be transmitted reliably (with asymptotically vanishing probability of error) through the channel with the help of a quantum encoder and decoder \cite{Watrous_QIT, Wildebook}. In the absence of a known single-letter formula for $C(T)$, it is useful for our purposes that
\begin{equation}\label{C-CH}
C(T) \geq C_{\text{H}}(T)\,,
\end{equation} 
as a consequence of the Holevo--Schumacher--West\-more\-land Theorem \cite{schumacher_sending_1997, holevo_capacity_1998}, where $C_{\text{H}}(T)$ denotes the \emph{Holevo quantity}, which is defined below in \eqref{eq:holevo capacity expression}. 

 The entanglement-assisted classical capacity, $C_{\text{E}}(T)$, is defined as the maximum asymptotic rate of such reliable information transmission with the help of an unlimited supply of pre-shared entanglement between sender and receiver. Contrary to $C(T)$, this quantity is known to be expressible in a concise way in terms of relative entropies as shown in \cite{Bennet_etAl_entanglementassisted_descript}, see \eqref{CEdef} below (see \cite[Chapter 8]{Watrous_QIT} for a thorough discussion of various types of capacities of quantum channels). 

Our main result is the following bound on the ratio $C_{\text{E}}(T)/C(T)$.

\begin{theorem*}  \label{main}
Let $T:\mathcal{B}(A) \rightarrow \mathcal{B}(B)$ be an arbitrary quantum channel. Then,
\begin{align*} 
C_{\text{E}} (T) \ \leq \ \frac {(4d_A -3)(2d_A -\frac 52)^2}{(2d_A -\frac 32)\ln(2d_A-\frac 32) -2d_A +\frac 52}  \  C(T)\,. 
\end{align*}
\end{theorem*}
For large $d_A$, the prefactor on the right-hand side scales as $\frac{8 d_A^2}{\ln(d_A)}(1+o(1))$. From the explicit capacity expressions from \cite{Holevo_conjectureChannel}, we can obtain an achievable prefactor which scales as $2 d_A (1+o(1))$ for large $d_A$ and any uniform upper bound on the ratio must be lower-bounded by that expression.

\section{Proof of the main result} \label{section:proving the bounds}

\subsection{Definitions} \label{subsec:capacity expressions}

The relative entropy of two states $\rho, \tau$ on the same space is defined as 
\begin{align} \label{eq:def relative entropy}
D \left( \rho \middle\| \tau \right) = \begin{cases} \Tr \left[ \rho\ln \rho - \rho\ln \tau \right] \;\; \mbox{if ${\rm ker}\,\tau\subseteq {\rm ker}\,\rho$}\\ \infty\;\qquad\mbox{\rm otherwise}\end{cases}\,.
\end{align}

We use the natural logarithm (instead of the conventional binary logarithm) in order to simplify the presentation of subsequent relative entropy bounds, but note that the main result is independent of this choice. 

Given a channel $T: \mathcal{B}(A)\to\mathcal{B}(B)$, its Holevo quantity, $C_{\text{H}}(T)$, can be expressed as \cite{Schumacher_signal_ensembles, Ohya_1997onQuantChannels, tomamichel_second-order_2015}
\begin{align} \label{eq:holevo capacity expression}
C_{\text{H}}(T) = \frac1{\ln 2}\inf_{\sigma_B}\sup_{\rho_A} D \left( T(\rho_A) \middle\| \sigma_B \right)\,,
\end{align}
where $\rho_A$ and $\sigma_B$ denote arbitrary quantum states on $\cS(A)$ and $\cS(B)$, respectively. W.l.o.g.\ $\sigma_B$ can be restricted to the channel image \cite{Schumacher_signal_ensembles}. 

The entanglement-assisted capacity of $T$ is defined in \cite{Bennet_etAl_entanglementassisted_descript} and shown to equal  
\begin{align} \label{CEdef}
& C_{\text{E}}(T) = \frac1{\ln 2} \sup_{\rho_{A'A}} D \big( (\id_{A'}\otimes T) \rho_{A'A} \big\| \rho_{A'} \otimes T(\rho_A) \big)\,,
\end{align} 
where $A'$ is an arbitrary auxiliary system with $d_{A'}=d_A$, $\rho_{A'A}$ is optimized over all pure states in $\cS_*(A'A)$, and we use the standard notation $\rho_A=\Tr_{A'} \left[  \rho_{A'A} \right]$, where $\Tr_{A'}$ denotes the partial trace w.r.t.\ $A'$.
The relative entropy expression corresponds to the mutual information of the channel $T$. We further note that
\begin{align}
\label{eq:proof bound:EAcapacity 1st upper bound}
    &D \left( (\id_{A'}\otimes T) \rho_{A'A} \middle\| \rho_{A'} \otimes T(\rho_A) \right) \notag\\
    &\quad \leq D( (\id_{A'}\otimes T) \rho_{A'A} \| \rho_{A'} \otimes \sigma_B)
\end{align}
for any state $\sigma_B$. In order to obtain a useful upper bound on $C_{\text{E}}$, we will make a judicious choice of $\sigma_B(\rho_{A'A})$, dependent on $\rho_{A' A}$, in \eqref{eq:proof bound:EAcapacity 1st upper bound} and apply some basic estimates of $ D ( \cdot \| \cdot )$ that we next proceed to establish.

\subsection{Relative entropy estimates} \label{subsec:preliminary entropy estimates}

When $\phi \in \mathcal{B}(A)$ is positive definite, we denote by $K_\phi$ the  quadratic form, defined on the self-adjoint part of $\cB(A)$, corresponding to $\nabla\ln\phi$, i.e. 
\begin{equation}\label{quadr}
K_\phi(\eta) = \left( \eta \,|\, \nabla_\eta\ln\phi \,\right) = \int_0^\infty dx\, \tr\, \big[\eta (\phi +x  \text{I}_A)^{-1}\big]^2\,,
\end{equation}
 where $\left( \cdot|\cdot \right)$ denotes the Hilbert-Schmidt inner product. See e.g. the End Matter section for elaborations on the equality above. 
 

In the proof of the main theorem, we shall make use of the following estimates for $D \left( \rho \middle\| \tau \right)$ in terms of $K_\tau(\rho-\tau)$


\begin{enumerate}[label=\upshape(\roman*), wide, labelindent=0pt]
\item\label{item:sUpper} For arbitrary $\rho,\tau\in \cS(A)$, where $\tau$ is positive definite, it holds that
\begin{equation}\label{eq: Supper} 
D \left( \rho \middle\| \tau \right) \, \leq \,  K_\tau(\rho-\tau)\,.
\end{equation}
\item\label{item:sLow} For arbitrary $\rho,\tau\in \cS(A)$, where $\tau$ is positive definite such that $k\tau\geq \rho$ for a constant $k \geq 1$, it holds that
\begin{align}\label{eq: Slow} 
D \left( \rho \middle\| \tau \right) \, \geq \, g(k) K_\tau(\rho-\tau)\,
\end{align}
where 
$$
g(k) = \frac{k\ln k -k +1}{(k-1)^2}\,.
$$
\end{enumerate}




The inequality in \ref{item:sUpper} is a well known consequence of the concavity of the logarithm, see e.g. \cite{Audenaert_2005_Rentropybounds}. For full-rank $\rho$, the proof of both estimates can be found in \cite[Lemma~2.2]{Gao_2022}. 
For the benefit of the reader, we provide proofs of both \ref{item:sUpper} and \ref{item:sLow} in the Appendix. 


For later reference we note that, if the eigenvalues of the state $\tau$ are $\mu_1\leq \dots\leq \mu_{d_A}$ with a corresponding orthonormal basis of eigenvectors  $|h_1\rangle,\dots, |h_{d_A}\rangle$, then  we have by \eqref{quadr} that
\begin{align}
 K_{\tau}(\eta) &=  \int_0^\infty dx\, \sum_{k,l=1}^{d_A} |\langle h_k\,| \eta |\, h_l\rangle|^2 (x + \mu_k)^{-1}(x + \mu_l)^{-1}\notag \\ \label{K-simple} & = 
 \sum_{k,l=1}^{d_A} |\langle h_k\,| \eta |\, h_l\rangle|^2 \frac{\ln \mu_k -\ln\mu_l}{\mu_k -\mu_l}
\end{align}
for any self-adjoint $\eta\in \cB(A)$, where the expression  $\frac{\ln \mu_k -\ln\mu_l}{\mu_k -\mu_l}$ is to be set equal to $\mu_k^{-1}$ in case $\mu_k=\mu_l$. 

\smallskip

If $\tau$ is not positive definite, i.e. if $\mu_k=0$ for some $k$ with notation as above, we may replace $\tau$ by $\tau +\epsilon \text{I}_A, \epsilon >0$, in the previous arguments  and observe that 
\begin{align*}
D \left( \rho \middle\| \tau \right) = \lim_{\epsilon\to 0} \Tr \left[ \rho\ln \rho - \rho\ln(\tau + \epsilon \text{I}_A) \right]  
\end{align*}
to conclude that the inequalities \eqref{eq: Supper} and \eqref{eq: Slow} still hold, with $K_\tau$ given by \eqref{K-simple} provided the conventions 
\begin{align*}
& |\langle h_k\,| \eta |\, h_l\rangle|^2 \frac{\ln \mu_k -\ln\mu_l}{\mu_k -\mu_l}   \\
& = \begin{cases} \infty\; \ \mbox{if $\langle h_k\,| \eta |\, h_l\rangle \neq 0$ and $\mu_k=0$ or $\mu_l=0$}\\ 0\quad \: \mbox{if $\langle h_k\,| \eta |\, h_l\rangle = 0$ and $\mu_k=0$ or $\mu_l=0$}\end{cases}
\end{align*}
are used. Thus, we can restrict the sum in (\ref{K-simple}) to include only pairs of non-zero eigenvalues of $\tau$ provided $\ker(\tau) \subseteq \ker(\rho)$ and still use the inequalities (\ref{eq: Supper}) and (\ref{eq: Slow}) in this case. This will be done without further notice in the remainder of this article.

\subsection{Proof of the bound on \texorpdfstring{$C_{\text{E}}/C$}{CEoverC}} \label{subsec:proof of simple bound}

As both sides of the stated inequality vanish if $d_A=1$, we can assume in the following that $d_A\geq 2$. 

Recalling (\ref{eq:proof bound:EAcapacity 1st upper bound}) and \eqref{eq:holevo capacity expression}, our goal is to establish an upper bound on $D \left( (\id_{A'}\otimes T)\rho_{A'A} \middle\| \rho_{A'}\otimes{\sigma_B(\rho_{A'A})} \right)$ for all $\rho_{A' A} \in \cS_*(A' A)$ and suitable $\sigma_B(\rho_{A'A})$ in terms of  $D \left( T(\rho'_A) \middle\|\sigma'_B \right),\,\sigma'_B\in \cS(B)$, for some $\rho'_A\in \cS(A)$. $\rho_{A'A}$ will by assumption be pure, i.e. $\rho_{A'A} = |v\rangle\langle v|$ for some $\ket{v} \in \mathcal{H}_{A' A}$, $\braket{v|v}=1$. 
Given a state $\sigma_B$, let the replacement channel $R_{\sigma_B}:{\mathcal B}(A) \to {\mathcal B}(B)$ be defined by 
\begin{align*}
R_{\sigma_B} (\phi) = {\rm tr}(\phi) \sigma_B\,,\quad \phi\in {\mathcal B}(A)\,.
\end{align*}
Then $(\id_{A'}\otimes R_{\sigma_B})(\rho_{A'A}) = \rho_A\otimes \sigma_B$ for any $\rho_{A'A}\in \cS({A'A})$.
Introducing  the notation $\Delta T = T-R_{\sigma_B}$, we have $\Delta T(\rho_A) = T(\rho_A) - \sigma_B$ and 
\begin{align*}
(\id_{A'}\otimes \Delta T)(\rho_{A'A}) = (\id_{A'}\otimes T)(\rho_{A'A}) - \rho_{A'}\otimes \sigma_B
\end{align*}
for any $\rho_{A'A}\in \cS({A'A})$. By \ref{item:sUpper}, we hence have
\begin{align}
& D \left( (\id_{A'}\otimes T)(\rho_{A'A}) \middle\| \rho_{A'}\otimes{\sigma_B} \right) \notag \\
& \qquad \leq K_{\rho_{A'}\otimes \sigma_B} \left( (\id_{A'}\otimes  \Delta T)(\rho_{A'A}) \vphantom{1^1} \right)\,.\label{est1}
\end{align}
 We next proceed to estimate the last expression in terms of quantities of the form $ K_{\sigma_B} (\Delta T(\rho'_A))$ for suitable pure states $\rho'_A\in \cS_*(A)$. As is well known, we can write $|v\rangle$ in the form $|v\rangle = \sum_{k=1}^{d_A} \alpha_k |e_k\rangle \otimes |f_k\rangle$, where $|e_1\rangle,\dots, |e_{d_A}\rangle$ and $|f_1\rangle,\dots, |f_{d_A}\rangle$ are orthonormal bases for $A'$ and $A$, respectively, and $\alpha_1\geq \dots\geq \alpha_{d_A}$ are non-negative real numbers whose squares sum to $1$. We then have
\begin{align}\label{rhoform}
\rho_{A'A} = \sum_{k,l=1}^{d_A} &\alpha_k\alpha_l |e_k\rangle\langle e_l|\otimes |f_k\rangle\langle f_l| \; , \\ \label{bigeta}
(\id_{A'}\otimes \Delta T)(\rho_{A'A}) & =  \sum_{k,l=1}^{d_A} \alpha_k\alpha_l |e_k\rangle\langle e_l|\otimes \Delta T(|f_k\rangle\langle f_l|)\,.
\end{align}
Moreover, $
\rho_{A'}= \sum_{k=1}^{d_A} \alpha_k^2 |e_k\rangle\langle e_k|$, which implies 
$$
\rho_{A'}\otimes \sigma_B =  \sum_{k=1}^{d_A}\sum_{r=1}^{d_B} \alpha_k^2 \mu_r\,|e_k\rangle\langle e_k|\otimes |g_r\rangle\langle g_r|\,,
$$
where the eigenvalues of $\sigma_B$ have been denoted by $\mu_1\geq\dots\geq \mu_{d_B} \geq 0$  with a corresponding orthonormal basis of eigenvectors $ \ket{g_1},\dots, \ket{g_{d_B}}$. Using this diagonal form of $\rho_{A'}\otimes\sigma_B$ in \eqref{K-simple} with the system $A$ replaced by $A'B$ and $\tau = \rho_{A'}\otimes\sigma_B$ and $\eta = (\id_{A'}\otimes \Delta T ) (\rho_{A'A})$ given by \eqref{bigeta}, we get 
\begin{align}
&K_{\rho_{A'}\otimes\sigma_B} ((\id_{A'}\otimes \Delta T)(\rho_{A'A}))\nonumber\\
& = \sum_{k,l,r,s} \frac{\ln\alpha_k^2\mu_r - \ln\alpha_l^2\mu_s}{\alpha_k^2\mu_r -\alpha_l^2\mu_s} \alpha_k^2\alpha_l^2 |\langle g_r\,| \Delta T(|f_k\rangle \langle f_l|)|\,g_s\rangle|^2 \nonumber \\ & \leq 
\sum_{k,l,r,s} \frac{\ln\mu_r - \ln\mu_s}{\mu_r -\mu_s} \max\{\alpha_k^2,\alpha_l^2\}|\langle g_r\,| \Delta T(|f_k\rangle \langle f_l|)|\, g_s\rangle|^2,
\label{mainest}
\end{align}
where the inequality follows by using  that $\ln$ is a concave function and hence the slope $\frac{\ln x -\ln y}{x-y}$ of the secant connecting the points $(x,\ln x)$ and $(y,\ln y)$ is a decreasing function of both $x$ and $y$.  

The states $\chi^a_{k,l} := \frac 12 \left( |f_k \rangle+ i^a  |f_l\rangle \right) \left( \langle f_k | +(-i)^a \langle f_l| \right)$, $k\neq l$ and $a=0,1,2,3$, are pure and are seen to  satisfy the identities $|f_k\rangle\langle f_l|+ |f_l\rangle\langle f_k| = \chi^0_{k,l} -\chi^2_{k,l} $
and $|f_k\rangle\langle f_l| - |f_l\rangle\langle f_k| = i(\chi^1_{k,l} -\chi^3_{k,l})$.
By linearity of $\Delta T$, these imply
\begin{align}
& \big|\langle g_r\,| \Delta T(|f_k\rangle \langle f_l|)|\, g_s\rangle\big|^2 + \big|\langle g_r\,| \Delta T(|f_l\rangle \langle f_k|)|\, g_s\rangle\big|^2 \notag \\ 
& \quad \leq \sum_{a=0}^3 \big|\langle g_r\,| \Delta T(\chi_{k,l}^a)| \, g_s\rangle\big|^2\,.\notag
\end{align}
Using this estimate in (\ref{mainest}) and recalling (\ref{K-simple}) we find that 
\begin{align}\label{est2}
 & K_{\rho_{A'}\otimes\sigma_B} ((\id_{A'}\otimes \Delta T)\rho_{A'A})  \leq \sum_{k} \alpha_k^2 K_{\sigma_B}(\Delta T(|f_k\rangle\langle f_k|))\nonumber\\ &~~~~~~~~+ \frac 1{2}\sum_{a=0}^3 \sum_{k\neq l} \max\{\alpha_k^2,\alpha_l^2\} K_{\sigma_B} (\Delta T(\chi_{k,l}^a))\,.
\end{align}
We shall choose $\sigma_B$ below such that there exists a $k_A\geq 1$, depending only on $d_A$, fulfilling 
\begin{align} \label{cond1}
k_A  \sigma_B\,\geq\, T(\chi_{k,l}^a)\quad\mbox{and}\quad k_A  \sigma_B\,\geq\, T(|f_k\rangle \langle f_k|)
\end{align}
for all $a$ and $k\neq l$. Hence, by applying \ref{item:sLow} to \eqref{est2} and using $\max\{\alpha_k^2,\alpha_l^2\}\leq \alpha_k^2+\alpha_l^2$, it follows that 
\begin{align}
& K_{\rho_{A'}\otimes\sigma_B}((\id_{A'} \otimes\Delta T)(\rho_{A'A}))  \notag \\
& \qquad \leq  g(k_A)^{-1}\Big( \sum_{k} \alpha_k^2 D(T(|f_k\rangle\langle f_k|)\, 
 \|\,\sigma_B) \notag \\
& \qquad + \frac 12 \sum_{a=0}^3 \sum_{k\neq l} (\alpha_k^2 + \alpha_l^2) D(T(\chi^a_{k,l})\, \|\,\sigma_B)\Big) \,,\label{est3}
\end{align}
holds for all $\rho_{A'A}$ given by \eqref{rhoform} and any state $\sigma_B$ fulfilling \eqref{cond1}.

The expression on the right-hand side of this inequality is a linear combination of relative entropies with positive coefficients which can be further estimated by applying the fact that, for any given convex combination $\rho = \sum_{i=1}^n p_i \rho_i$ of states $\rho_1,\dots,\rho_n \in \cS(B)$, the quantity $\sum_{i=1}^n p_i D \left( \rho_i \middle\| \sigma \right)$ as a function of $\sigma \in \cS(B)$ assumes a unique minimum at $\sigma=\rho$. This  follows immediately from the identity \cite{donald_further_1987} 
\begin{align*} 
\sum_{i=1}^n p_i D \left( \rho_i \middle\| \sigma \right) = \sum_{i=1}^n p_i D \left( \rho_i \middle\| \rho \right) + D \left( \rho \middle\| \sigma \right)
\end{align*}
that is straight-forward to verify, and often called \emph{Donald's identity}. Choosing $\sigma_B(\rho_{A'A})$ so as to minimise the right-hand side of \eqref{est3}, we thus have   
\begin{align} \label{sigmaB}
&M_A\sigma_B(\rho_{A'A}) \notag\\ & = \sum_{k} \alpha^2_k T(|f_k\rangle\langle f_k|) + \frac 12 \sum_{a=0}^3 \sum_{k\neq l} (\alpha^2_{k} +\alpha^2_l)T(\chi^{a}_{k,l})\notag\\
& = \sum_{k} \alpha^2_k T(|f_k\rangle\langle f_k|) + 2 \sum_{k\neq l} (\alpha^2_{k} + \alpha^2_l)T(|f_k\rangle\langle f_k|)\,, 
\end{align}
where we have used that $\chi^0_{k,l} + \chi^2_{k,l} = \chi^1_{k,l} + \chi^3_{k,l} = |f_k\rangle\langle f_k| + |f_l\rangle\langle f_l|$, and $M_A$ is a normalization factor given by 
\begin{align*}
M_A &= \sum_{k} \alpha_k^2 + 2\sum_{k\neq l} (\alpha_k^2 + \alpha_l^2) = 4d_A -3 \,,
\end{align*}
where the last equality follows by recalling that  $\alpha_1^2+\dots+\alpha_{d_A}^2=1$. Moreover, from \eqref{sigmaB} we get  that 
\begin{align*}
& M_A \sigma_B(\rho_{A'A}) \\ & = 2\sum_k T(|f_k\rangle\langle f_k|) + (2d_A-3)\sum_k \alpha_{k}^2 T(|f_k\rangle\langle f_k|)\\ &  \geq 2 T(I_A)\,,\notag 
\end{align*}
from which it follows that \eqref{cond1} holds for $\sigma_B(\rho_{A'A})$ with $k_A = \frac 12 M_A = 2d_A-\frac 32$, as $T(I_A)\geq T(\rho_A)$ for any state $\rho_A$. In conjunction with the definition of $\sigma_B(\rho_{A'A})$ this shows by \eqref{est3} that 
\begin{align*}
& K_{\rho_{A'}\otimes \sigma_B(\rho_{A'A})}((\id_{A'}\otimes \Delta T)(\rho_{A'A})) \\
& \leq \frac{M_A}{g(\frac 12 M_A)} \Big( \sum_{k} p_k D \left( T(|f_k\rangle\langle f_k|) \middle\| \sigma_B(\rho_{A'A}) \right)  \\
& \qquad \quad + \sum_{a=0}^3 \sum_{k \neq l} p^a_{k,l}  D \left( T(\chi^{a}_{k,l}) \middle\| \sigma_B(\rho_{A'A}) \right)\Big)  \\
&  \leq \frac{M_A}{g(\frac 12 M_A)} \Big( \sum_{k} p_k D \left( T(|f_k\rangle\langle f_k|) \middle\| \tau_B \right)  \\
& \qquad \quad +  \sum_{a=0}^3 \sum_{k \neq l} p^a_{k,l} D \left( T(\chi^{a}_{k,l}) \middle\| \tau_B \right)\Big) 
\end{align*}
for all states $\tau_B\in \cS(B)$, where we introduced the probability distribution $\{p_k\}_k\cup \{p_{k, l}^a\}_{k \neq l}^a$ by normalizing the $\alpha$-dependent coefficients with $M_A$.
Estimating $ D(T(\chi^{a}_{k,l})\,\|\, \tau_B)$ by $\sup_{\omega_A\in \cS(A)}  D(T(\omega_A)\,\|\,\tau_B)$, we obtain 
\begin{align} 
 & K_{\rho_{A'}\otimes \sigma_B(\rho_{A'A})}((\id_{A'}\otimes \Delta T)(\rho_{A'A})) \notag \\ \label{est5}
 & \qquad \quad \leq  \frac{M_A}{g(\frac 12 M_A)} \sup_{\omega_A\in \cS(A)} D \left( T(\omega_A) \middle\|  \tau_B \right) \,.
\end{align}
Finally, recalling the definition (\ref{eq:holevo capacity expression}) of $C_{\text{H}}(T)$ and using (\ref{est1}), we conclude by taking the infimum over $\tau_B \in \cS(B)$ on the right-hand side of \eqref{est5} that  
\begin{align*}
D \left( (\id_{A'}\otimes T)(\rho_{A'A}) \middle\| \rho_{A'}\otimes\sigma_B(\rho_{A'A}) \right) \\
\leq  \frac{M_A \ln 2}{g(\frac 12 M_A)}  C_{\text{H}}(T)
\end{align*}
holds for all $\rho_{A' A} \in \cS_*(A' A)$. By evaluating $g(\frac 12 M_A)$ this proves the theorem as a consequence of \eqref{CEdef}, \eqref{eq:proof bound:EAcapacity 1st upper bound} and \eqref{C-CH}. 

Note that $C(T)$ can be replaced with $C_{\text{H}}(T)$ in the theorem.

\section{Conclusion} \label{section:applications}

In realistic quantum communication, the encoders and decoders of sender and receiver are subject gate-level noise. Considering a fixed level of gate noise $p > 0$ independent of the block length of the codes leads to new notions of capacity, in particular raising the question whether the traditional capacity is approached for small $p$ \cite{Christandl_fault_tolerant_f_Qcomm}. 

In this context, the bound 
$$C_{\text{E}}(T) \geq C_{\text{E}}(T, p) \geq C_{\text{E}}(T)-f(p) \frac{C_{\text{E}}(T)}{C(T)},$$ 
has been obtained for the fault-tolerant entanglement-assisted capacity $C_{\text{E}}(T, p)$ \cite{BelzigIEEE} (see also \cite{Belzig23}). 
The function $f(p)$ goes to zero as $p$ goes to zero and depends only on the dimension of the channel. This continuity bound is therefore a priori channel-dependent and it was left open in \cite{BelzigIEEE} whether a bound uniform for channels in a given dimension can be obtained. The  uniform upper bound on $C_{\text{E}}(T)/C(T)$ presented in this work now implies such a uniform bound.

A question for future investigations is whether one can improve our $O(d_A^2/\ln d_A) $ bound to $O(d_A)$, which would be tight in view of the existing lower bounds. 

We also leave open whether there exists a bound on the ratio of capacities which only depends on the output dimension, or conversely, whether there exists a series of quantum channels with fixed output dimension (and necessarily, as a consequence of our work, with increasing input dimension) for which the ratio diverges.

\begin{acknowledgments}
We would like to thank Alexander M{\"u}ller-Hermes and Andreas Winter for helpful discussions. We would also like to thank Christoph Hirche for pointing out reference \cite{Gao_2022} to us.

MC, BD and LHW acknowledge financial support from VILLUM FONDEN via the QMATH Centre of Excellence (Grant No.10059). MC also thanks the European Research Council (ERC Grant Agreement No. 818761) and the Novo Nordisk Foundation (grant NNF20OC0059939 ‘Quantum for Life’). PB acknowledges financial support  from the Canada First Research Excellence Fund. MT is supported by the National Research Foundation, Singapore and A*STAR under its CQT Bridging Grant.

\end{acknowledgments}

\bibliography{references}


\subsection*{Appendix: Proof of relative entropy inequalities}

Before proving the relative entropy estimates \ref{item:sUpper} and \ref{item:sLow}, we shall derive a useful integral representation for the relative entropy $D(\rho\,\|\,\tau)$, relating it to the quadratic form $K_\phi$ given in (\ref{quadr}).

Making use of the well-known integral representation of the matrix logarithm
\begin{align*}
\ln \phi = \int_0^\infty dx\,\Big((1 +x)^{-1}\text{I}_A - (\phi + x \text{I}_A )^{-1}\Big)\,,
\end{align*}
together with the resolvent identity
\begin{align*} &
(\phi +x \text{I}_A )^{-1} - (\psi +x \text{I}_A )^{-1} \\ 
&=  (\phi +x \text{I}_A )^{-1}(\psi-\phi)(\psi +x \text{I}_A )^{-1}\,,
\end{align*}
it follows that
\begin{align} \label{eq:logarithm difference}
\ln \psi - \ln \phi = \int_0^\infty dx\,(\phi + x \text{I}_A )^{-1}(\psi-\phi)(\psi + x \text{I}_A )^{-1} 
\end{align}
for any positive definite operators $\phi, \psi$ in $\cB(A)$.
Setting $\psi = \phi + t\cdot\eta,\, t\in\mathbb R$, in \eqref{eq:logarithm difference} and differentiating w.r.t.\ $t$ at $t=0$, it follows that the derivative of $\ln $ at $\phi$ in the direction $\eta\in \cB(A)$ is given by 
\begin{align*} 
\nabla_{\eta}\ln (\phi) = \int_0^\infty dx \, (\phi + x \text{I}_A )^{-1}\eta (\phi + x \text{I}_A )^{-1} \,.
\end{align*}
Thus, given any smooth curve $\rho_t,\,0\leq t\leq 1$, of positive definite states in $\cS(A)$  such that $\rho_0=\tau$ and $\rho_1 = \rho$, we have 
\begin{equation}\nonumber
\ln \rho - \ln \tau = \int_0^1 dt\,\int_0^\infty dx\,(\rho_t + x \text{I}_A)^{-1}\dot\rho_t(\rho_t + x \text{I}_A)^{-1} \,,
\end{equation}
where $\dot\rho_t$ denotes the derivative of $\rho_t$. In particular, setting $\rho_t = t\rho + (1-t)\tau$, we have $\dot\rho_t=\rho-\tau$, and using this expression in the definition of $D(\rho\,\|\,\tau)$ we obtain 
\begin{align}
D(\rho&\,\|\,\tau) = \nonumber\\
&\int_0^1 dt \int_0^\infty dx\, \tr\,\rho (\rho_t +x \text{I}_A)^{-1}(\rho - \tau)(\rho_t + x \text{I}_A)^{-1}\,.\nonumber
\end{align}
Writing the first factor of the integrand as $\rho = (1-t)(\rho-\tau) + \rho_t$ and using  
\begin{align}
& \int_0^1 dt \int_0^\infty dx\, \tr\, \rho_t (\rho_t +x \text{I}_A)^{-1}(\rho - \tau) (\rho_t +x \text{I}_A)^{-1} \nonumber\\
& = \int_0^1 dt \int_0^\infty dx\, \tr\, \rho_t (\rho_t +x \text{I}_A)^{-2}(\rho - \tau) \nonumber\\ 
& = \int_0^1 dt \Big[-\tr\,\rho_t (\rho_t + x \text{I}_A)^{-1}(\rho-\tau)\Big]_{x=0}^{x=\infty}\nonumber\\ 
& = \int_0^1 dt \ \tr\,\rho_t\rho_t^{-1}(\rho-\tau)  = 0\,,\nonumber
\end{align}
since  $\tr\,\rho = \tr\,\tau =1$, it follows that 
\begin{equation}\label{Sintrep1}
D(\rho\,\|\,\tau) = \int_0^1 dt\,(1-t) \int_0^\infty dx\, \tr \big[(\rho-\tau) (\rho_t +x \text{I}_A)^{-1} \big]^2\,.
\end{equation}
Making use of the quadratic form $K_\phi$ defined in \eqref{quadr}, the integral representation
\eqref{Sintrep1} can be written as 
\begin{equation}\label{Sintrep2}
D \left( \rho \middle\| \tau \right) = \int_0^1dt\,(1-t) K_{\rho_t}(\rho-\tau)\,.
\end{equation}
We further note that $K_\phi$ is monotonically decreasing in $\phi$ in the sense that 
\begin{equation}\label{Kmon}
K_\phi(\eta) \geq K_\psi(\eta) \quad\mbox{if $\phi\leq \psi$}\,,
\end{equation}
for arbitrary positive definite operators $\phi,\psi$ and any self-adjoint operator $\eta$ in $\cB(A)$. 
The proof of this property makes use of the well known fact that $\tau^{-1}$ is a decreasing operator function of $\tau> 0$ and that $\eta\tau\eta\geq 0$, whenever $\tau\geq 0$ and $\eta$ is self-adjoint.
Using positivity and cyclicity of the trace, the latter property implies, for arbitrary non-negative operators $\tau$ and $\zeta$, that
$$\tr\,\zeta\tau = \tr\,\tau^{\tfrac 12}\zeta\tau^{\tfrac 12}\geq 0\,,$$
and hence $\tr\,\zeta\tau$ is an increasing function of both $\tau\geq 0$ and $\zeta\geq 0$.
Using this fact, we get, for $0\leq\phi\leq \psi$ and $x>0$,
\begin{align*}
& \quad\;\tr \ \eta (\phi+x \text{I}_A)^{-1} \eta (\phi+x \text{I}_A)^{-1} \\ 
 & \geq \tr \ \eta (\psi+x \text{I}_A)^{-1} \eta (\psi+x \text{I}_A)^{-1} \,,
\end{align*}
which together with \eqref{quadr} proves the claimed monotonicity.

We are now ready to prove the relative entropy estimates \ref{item:sUpper} and \ref{item:sLow}, which are reproduced below. 



\hfill
\begin{enumerate}[label=\upshape(\roman*), wide, labelindent=0pt]
\item\label{app:item:sUpper} For arbitrary $\rho,\tau\in \cS(A)$, where $\tau$ is positive definite, it holds that
\begin{equation*}
D \left( \rho \middle\| \tau \right) \, \leq \,  K_\tau(\rho-\tau)\,.
\end{equation*}
\item \label{app:item:sLow} For arbitrary $\rho,\tau\in \cS(A)$, where $\tau$ is positive definite such that $k\tau\geq \rho$ for a constant $k \geq 1$, it holds that
\begin{align*}
D \left( \rho \middle\| \tau \right) \, \geq \, g(k) K_\tau(\rho-\tau)\,
\end{align*}
where $g(k) = \frac{k\ln k -k +1}{(k-1)^2}\,$.
\end{enumerate}

\begin{proof}[Proof of \ref{item:sUpper} and \ref{item:sLow}]
Suppose first that $\rho$ is positive definite. Statement \ref{item:sUpper} then follows by using that $\rho_t \geq (1-t)\tau$ and hence $K_{\rho_t}(\eta)\leq K_{(1-t)\tau}(\eta)$ by \eqref{Kmon} which by use of \eqref{Sintrep2} gives 
\begin{align*}
&D(\rho\,\|\,\tau)   \\ & \leq \int_0^1(1-t) dt\int_0^\infty dx\,\,\tr \big[(\rho-\tau)((1-t)\tau+x \text{I}_A)^{-1}\big]^2 \\
& \quad = K_{\tau}(\rho-\tau)\,,
\end{align*}
where the last equality follows by a change of variable $x=(1-t)y$, whereby the $t$-dependence of the integrand disappears.

Statement \ref{item:sLow} follows similarly by using $\rho_t\leq (kt+1-t)\tau$. A change of variable $x= (kt+1-t)y$ then gives the claimed lower bound with
$$
g(k) = \int_0^1dt\, \frac{1-t}{kt+1-t} = \frac{k\ln k -k +1}{(k-1)^2}\,.
$$

In case $\rho$ has zero eigenvalues, the inequalities follow as both sides are continuous functions of $\rho$. 
\end{proof}

\end{document}